\documentclass[a4paper,reqno,11pt]{amsart}
\usepackage[a4paper,bindingoffset=0.2in,%
left=1in,right=1in,top=1in,bottom=1in,%
           footskip=.25in]{geometry}
\usepackage{graphicx}

\usepackage{amssymb}
\usepackage{amsthm}
\usepackage{amsmath}
\usepackage{color}
\usepackage[utf8]{inputenc}
\usepackage{tikz}
\usepackage[all]{xy}
\usepackage[bookmarksnumbered,colorlinks]{hyperref}
\usepackage{dsfont}
\usepackage[normalem]{ulem}
\usepackage{latexsym}
\usepackage{verbatim}
\usepackage{float} 

\usepackage[euler-digits]{eulervm}
\usepackage{enumerate}
\def\br#1\er{\textcolor{red}{#1}} %
\newcommand{\be}{\begin{equation}}
\newcommand{\ee}{\end{equation}}

\newcommand{\ben}{\begin{enumerate}}
\newcommand{\een}{\end{enumerate}}
\newcommand{\bit}{\begin{itemize}}
\newcommand{\eit}{\end{itemize}}
\newcommand{\edoc}{\end{document}}

\newcommand{\qcd}{\begin{flushright} $\Box$ \end{flushright}}

\author{}

\date{}
\title{Omniscient foliations and the geometry of cosmological spacetimes}

\author[I.P. Costa e Silva]{Ivan P. Costa e Silva}
 \address{Department of Mathematics,
	Universidade Federal de Santa Catarina,
	\hfill\break\indent 88.040-900 Florian\'{o}polis-SC, Brazil.}
\email{pontual.ivan@gmail.com
}

 \author[J.L. Flores]{Jos\'e L. Flores}
\address{Departamento de \'Algebra, Geometr\'{\i}a y Topolog\'{\i}a,  Universidad de M\'alaga
	\hfill\break\indent
	Facultad de Ciencias, Campus de Teatinos,
	\hfill\break\indent 29080 M\'alaga, Spain}
\email{floresj@uma.es}

 \author[J. Herrera]{J\'onatan Herrera}
\address{Departamento de Matem\'aticas, Universidad de C\'ordoba
	\hfill\break\indent
	Edificio Albert Einstein, Campus de Rabanales,
	\hfill\break\indent 14071 C\'ordoba, Spain}
\email{jherrera@uco.es}

\begin{document}
\newtheorem{thm}{Theorem}[section]
\newtheorem{teo}[thm]{Theorem}
\newtheorem{prop}[thm]{Proposition}
\newtheorem{lemma}[thm]{Lemma}
\newtheorem{cor}[thm]{Corollary}
\newtheorem{conv}[thm]{Convention}
\newtheorem{defi}[thm]{Definition}
\newtheorem{notation}[thm]{Notation}
\newtheorem{exe}[thm]{Example}
\newtheorem{conj}[thm]{Conjecture}
\newtheorem{prob}[thm]{Problem}
\newtheorem{rem}[thm]{Remark}
\maketitle
\usetikzlibrary{matrix}
\date{}

\begin{abstract}
We identify certain general geometric conditions on a foliation of a spacetime $(M,g)$ by timelike curves that will impede the existence of null geodesic lines, especially if $(M,g)$ possesses a compact Cauchy hypersurface. The absence of such lines, in turn, yields well-known restrictions on the geometry of cosmological spacetimes, in the context of Bartnik's splitting conjecture. Since the (non)existence of null lines is actually a conformally invariant property, such conditions only need to apply for some suitable conformal rescaling of $g$. 
\end{abstract}
\section{Introduction}


Let $(M,g)$ be an $(n+1)$-dimensional ($n\geq 1$) spacetime, and let $\mathcal{F}$ be a smooth foliation of $M$ whose leaves are timelike curves. The problem we propose to address in this note is: are there natural geometric conditions on $\mathcal{F}$ which entail that there are no {\it null (geodesic) lines} on $(M,g)$? 

Recall that a {\em causal geodesic line} is an inextendible causal geodesic maximizing the Lorentzian distance between any two of its points. A causal geodesic line can be either timelike or null, and the latter case occurs if and only if its image is achronal in $(M,g)$. The celebrated Lorentzian splitting theorem (see \cite[Ch. 14]{beembook} and references therein for a detailed account) states that if $(M,g)$ $(i)$ is either globally hyperbolic or timelike geodesically complete, $(ii)$ obeys the {\em timelike convergence condition} (TCC) - $Ric(v,v)\geq 0, \forall v\in TM$ timelike - and $(iii)$ contains a complete {\it timelike} line, then $(M,g)$ is isometric to a product spacetime $(\mathbb{R}\times \Sigma, -dt^2 \oplus h_0)$, where $(\Sigma, h_0)$ is a complete Riemannian manifold. (We refer to the latter situation in abbreviated form by saying that $(M,g)$ {\it splits}. Note that, in this case, $(M,g)$ ends up being both globally hyperbolic {\it and} causally geodesically complete \cite[Theorem 3.67]{beembook}.) 


A natural arena to apply the Lorentzian splitting theorem is the geometry of cosmological spacetimes. Recall that the spacetime $(M,g)$ is said to be {\it cosmological} if it is globally hyperbolic with compact Cauchy hypersurfaces and the $TCC$ holds. A famous open conjecture due to Bartnik \cite{bartnikoriginal} states that cosmological spacetimes should either be timelike geodesically incomplete or else split as above. 

There have been many different attempts to either prove or find a counterexample to Bartnik's conjecture, and partial results have been obtained under additional requirements \cite{bartnikoriginal,bartnikus,E,G1,G2,GL,GV1,GV2}. Now, a standard argument using the compactness of the Cauchy hypersurfaces and the Limit Curve Lemma \cite[Cap. 8]{beembook} establishes that every cosmological spacetime contains a causal geodesic line. Bartnik's conjecture would thus be an immediate corollary to the Lorentzian splitting theorem if one could prove that timelike complete cosmological spacetimes have no null geodesic lines, and this is our main motivation here. This issue is subtle because there are examples \cite{EhG} of globally hyperbolic spacetimes with compact Cauchy hypersurfaces which have {\it only} null lines. However, these do not obey the TCC. Thus, if the conjecture is true, then the TCC is a crucial controlling piece. Indeed, that Bartnik's conjecture is false without the TCC is immediately realized by considering de Sitter's spacetime.  

Of course, because a spacetime of the form $(\mathbb{R}\times \Sigma, -dt^2 \oplus h_0)$ with compact $\Sigma$ does not have any null lines, any geometric condition in cosmological spacetimes that implies Bartnik's splitting a fortiori also implies there are no null lines. But geometric conditions in cosmological spacetimes which either directly rule out null lines or establish the existence of a timelike line are known \cite{EhG,G1,bartnikus}. We are especially interested here in the {\it no observer-horizon} (NOH) condition \cite{G1}, namely, that for {\it every} inextendible timelike curve $\gamma$ in a spacetime $(M,g)$ we have 
\begin{equation}\label{noh1}
    I^{\pm}(\gamma)=M.
\end{equation}
It is easy to verify that the NOH condition in a strongly causal spacetime $(M,g)$ implies that the latter contains no null lines \cite{bartnikus}. Our interest in this particular condition is that it has been established that the existence of a {\it complete} timelike conformal Killing vector field (that is, $X\in \mathfrak{X}(M)$ such that $\mathcal{L}_Xg = \sigma \cdot g$, where $\sigma \in C^{\infty}(M)$ and $\mathcal{L}$ denotes the Lie derivative) in a cosmological spacetime does imply the NOH \cite{bartnikus}. We were thus led to wonder if this might hold for other kinds of timelike vector fields.  

In Ref. \cite{harrisomni}, Harris and Low considered the setting we started out this Introduction with: they called a smooth foliation $\mathcal{F}$ of $M$ by timelike curves {\it future omniscient} [resp. {\it past omniscient}] if for each leaf $\gamma$ of $\mathcal{F}$ we have $I^{-}(\gamma) = M$ [resp. $I^{+}(\gamma) =M$]. The foliation $\mathcal{F}$ then is said to be {\it omniscient} if it is both future and past omniscient; in other words (\ref{noh1}) applies to each leaf. (In physical terms, if we think of the leaves of $\mathcal{F}$ as describing the worldlines of a family of observers, then (\ref{noh1}) means that the whole universe can both receive and send signals to the observer $\gamma$.) The existence of an omniscient timelike foliation $\mathcal{F}$ is a weaker requirement than the NOH condition, because the latter implies that {\it any} timelike foliation is omniscient. The bridge with the approach in \cite{bartnikus} lies in that Harris and Low prove \cite[Corollary 3.2]{harrisomni} that if $(M,g)$ is chronological, then the integral curves of a complete timelike conformal Killing vector field define an omniscient timelike foliation. As we saw above, in a cosmological spacetime this would entail the nonexistence of null lines we seek. 

This work focuses on obtaining various conditions that ensure that a given timelike foliation is omniscient. As matter of fact, it turns out to be enough that it is either future or past omniscient for our purposes here: since in cosmological spacetimes the existence of a future/past omniscient foliation is actually equivalent to the future/past NOH condition [cf. Proposition \ref{omnitonoh} below], we thereby show there are no null lines. Observe, finally, that past/future omniscience, NOH condition or the (non)existence of null lines are {\it conformally invariant} properties; thus, when establishing such properties, one has the freedom to perform conformal rescalings of $g$ whenever convenient.

The rest of the paper is organized as follows. In section \ref{sec2} we study the NOH condition under the presence of a very general timelike foliation by applying notions from Finslerian Geometry. In section \ref{sec2.2} we particularize previous study to the case where the splitting provided by the timelike foliation is orthogonal. We will show that this allows a tighter study of the omniscience condition in the context of Bartnik conjecture. In section \ref{sec2.3} we apply and illustrate the concepts in previous sections to the important case of a gradient timelike vector field. Finally, in section \ref{sec2.4}, we show the versatility of our approach by showing its applicability to a general situation where a timelike foliation is not explicitly given.

\section{Finslerian geometry of timelike foliations and the NOH condition}\label{sec2}

Henceforth we fix a spacetime $(M,g)$ of dimension $n+1\geq 2$. All functions and (sub)ma\-ni\-folds are assumed to be smooth ($C^{\infty}$) unless otherwise stated, although we sometimes repeat that for emphasis. The reader is assumed to have familiarity with the basic tenets of Lorentzian geometry and causal theory of spacetimes as laid out in the core references \cite{beembook,oneill}. 

Let $\mathcal{F}$ be a smooth foliation of $M$ by timelike curves. Since $(M,g)$ is time-oriented, there always exists a future-directed timelike vector field $X\in \mathfrak{X}(M)$ whose integral curves coincide with the leaves of $\mathcal{F}$, in which case we say that $X$ {\it spans} $\mathcal{F}$. If $(M,g)$ is chronological, then $X$ can always be chosen to be complete \cite[Lemma 1.1]{harrisomni}. 

For the rest of this section, therefore, we fix a {\it complete} future-directed timelike vector field $X\in \mathfrak{X}(M)$. As we mentioned above, any timelike foliation can be thought of as being spanned by one such provided $(M,g)$ is chronological. 

We first introduce some notation. Let $\phi:\mathbb{R}\times M\rightarrow M$ be the flow of $X$, which defines a smooth $\mathbb{R}$-action on $M$. We shall need a {\it slice} of this action, namely a smooth hypersurface $\Sigma \subset M$ which is intersected exactly once by each integral curve of $X$, i.e., by each orbit of the action $\phi$. It turns out that such a slice exists if and only if $\phi$ is a {\it proper} action \cite[Proposition 2.11]{JLyo}. It is then important to know when $\phi$ is proper. It is evident that if $(M,g)$ is globally hyperbolic, then any smooth Cauchy hypersurface is a slice, irrespective of $X$. More generally, Harris and Low \cite{harrisomni} introduced the notion of an {\it ancestral pair} for the action $\phi$: a pair of points $p,q\in M$ such that the orbit of $p$ by $\phi$ is contained in $I^{-}(q)$. Thus, if $(M,g)$ is chronological and $\phi$ has no ancestral pairs, then $\phi$ is proper \cite[Theorem 1.2]{harrisomni} (see also \cite[Theorem 3.8]{JLyo}). 

The next result highlights an important relationship between omniscience and the NOH condition. 
\begin{prop}\label{omnitonoh}
Suppose $(M,g)$ is a strongly causal spacetime, and let $X \in \mathfrak{X}(M)$ be a complete timelike vector field whose flow admits a compact slice $\Sigma$. Assume also that $X$ spans a future (resp. past) omniscient timelike foliation $\mathcal{F}$. Then, the future (resp. past) NOH condition holds; hence, $(M,g)$ has no null lines. 
\end{prop}
\begin{proof}
Let us focus on the future direction, as the past case will be completely analogous. Consider $X$ a complete future-directed timelike vector field spanning the foliation $\mathcal{F}$; and let $\Sigma \subset M$ be a compact slice. Previous foliation allows us to define the difeomorphism $F:\mathbb{R}\times\Sigma\rightarrow M$, $F(t,x):=\gamma_x(t)$, where $\gamma_x$ is the integral curve of $X$ with $\gamma_x(0)=x$. Endow $\mathbb{R}\times\Sigma$ with $F^{*}g$ to make $F$ a Lorentzian isometry. 

Take a point $p\in M$ and consider $F^{-1}(p)=(t,x)\in \mathbb{R}\times\Sigma$. From the future omniscience condition for $\mathbb{R}\times\Sigma$ and the open character of the chronology, there exist for each $x'\in \Sigma$ some value $t_{x'}\in \mathbb{R}$ and some neighborhood $U_{x'}\ni x'$ such that $[t_{x'},\infty)\times U_{x'}\subset I^{+}(t,x)$.  

From construction and the strong causality, any future-inextendible timelike curve $\gamma$ in $(M,g)$ admits an exhaustive sequence ${s_k}$ of values of its domain of definition such that $\{\gamma(s_k)\}\subset M$ is not contained in any compact set. So, the same happens for $F^{-1}(\gamma(s_k))=(t_k,x_k)\subset \mathbb{R}\times\Sigma$. Moreover, up to a subsequence, $x_k\rightarrow x_{*}\in \Sigma$. Hence, $$F^{-1}(\gamma(s_k))=(t_k,x_k)\in [t_{x_{*}},\infty)\times U_{x_{*}}\subset I^{+}(t,x)=I^{+}(F^{-1}(p))\quad\hbox{for $k$ big enough.}$$ By composing with $F$, we deduce that $\gamma(s_k)\subset I^{+}(p)$ for $k$ big enough, as required. 
\end{proof}


In any case, we shall assume for the rest of this section that $\phi$ is proper, and fix a slice $\Sigma$. In this case, the map 
\begin{equation}
\label{eq:1}
\phi_{\Sigma}:= \phi|_{\mathbb{R}\times \Sigma}:\mathbb{R}\times \Sigma \rightarrow M
\end{equation}
is a diffeomorphism. Therefore, $\phi_{\Sigma}$ becomes an isometry if we consider the pullback metric $\phi_{\Sigma}^{*}g$ on $\mathbb{R}\times \Sigma$. Consider the smooth one-parameter family of 1-forms $\omega_{\tau}$ and a smooth one-parameter family of $(0,2)$-tensors $h_{\tau}$ on $\Sigma$ defined via the equations 
\begin{equation}\label{1formmetricinduced}
    \begin{array}{cc}
     \omega_{\tau}(V) := \phi_{\Sigma}^{*}g(\partial _{\tau}, \widetilde{V})   &  h_{\tau}(V,W) :=\phi_{\Sigma}^{*}g(\widetilde{V},\widetilde{W}),
        
    \end{array}
\end{equation}
where $V,W \in \mathfrak{X}(\Sigma)$, the tilde indicates their lift to $\mathbb{R}\times \Sigma$ through the canonical projection $\pi_2:\mathbb{R}\times \Sigma \rightarrow \Sigma$. Moreover, $\partial_{\tau}$ denotes the lift of the canonical vector field $d/dt$ on $\mathbb{R}$ through the canonical projection $\pi_1:\mathbb{R}\times \Sigma \rightarrow \mathbb{R}$, so that $X$ is the pushforward of $\partial_{\tau}$ through $\phi_{\Sigma}$. Finally, if we define the positive function $\beta \in C^{\infty}(\mathbb{R}\times \Sigma)$ by 
$$\beta := -\phi_{\Sigma}^{*}g(\partial _{\tau}, \partial_{\tau}),$$
we have
\begin{equation}
  \label{eq:3}
  \phi_{\Sigma}^{*}g=-\beta^2\, d\tau^{2} + (\pi_2^{*}\omega_{\tau})\otimes d\tau + d\tau \otimes (\pi_2^{*}\omega_{\tau}) + \pi_2^{*}h_{\tau}.
\end{equation}

\medskip

It is important to realize that, for a given $\tau$, the hypersurface $\{\tau \}\times \Sigma$ does not have to be spacelike in $(\mathbb{R}\times \Sigma,\phi_{\Sigma}^{*}g)$, and hence $h_{\tau}$ may not be positive-definite or even nondegenerate. However, as $g$ is a Lorentzian metric, we can define a $1$-parameter family of Riemannian metrics on $\Sigma$ by considering $\rho_{\tau} := h_{\tau} + \frac{\omega_{\tau}}{\beta}\otimes \frac{\omega_{\tau}}{\beta}$ (the fact that $\rho_{\tau}$ is Riemannian follows from the same arguments in \cite[Proposition 3.1]{HJ})). Even more, from $\omega_{\tau}$ and $h_{\tau}$ we can construct a one-parameter family of {\it pre-Finsler metrics} on $\Sigma$ with a strong connection with the causality of $(M,g)$. Let us recall the notion of pre-Finsler metric (see \cite{HJ} for details):

\begin{defi}
	\label{def:Ivan-General-Retocado-Jony:1}
	Let $\Sigma$ be a $n$-dimensional manifold and $T\Sigma$ its tangent bundle. We will call a \textrm{pre-Finsler metric} to a continuous map $F:T\Sigma\rightarrow \mathbb{R}$ satisfying:
	\begin{itemize}
		\item[(i)] $F:T\Sigma\setminus \left\{ \mathbf{0} \right\}\rightarrow \mathbb{R}$ is smooth, where $\mathbf{0}$ denotes the zero section.
				\item[(ii)] $F$ is positive homogeneous, that is, $F(\lambda v)=\lambda F(v)$ for all $v\in T\Sigma$ and $\lambda>0$.
	\end{itemize}
	The pair $(\Sigma,F)$ will be called a \textrm{pre-Finsler manifold}.
\end{defi}

Now observe that, following the ideas developed in \cite{CJS11,FHS13}, a curve $\gamma(s)=(\tau(s),x(s))$ is a future-directed (resp. past-directed) causal curve if $F^{+}_{\tau(s)}(\dot{x}(s))\leq\beta(\tau(s),x(s))\dot{\tau}(s)$ (resp. $F^{-}_{\tau(s)}(\dot{x}(s))\geq\beta(\tau(s),x(s))\dot{\tau}(s)$), where the pre-Finsler metrics $F^{\pm}_{\tau}:T\Sigma \rightarrow \mathbb{R}$ are given by 
\begin{eqnarray}\label{randerseq}
	F^{\pm}_{\tau}(v) &:=& \pm \frac{\omega_{\tau}(v)}{\beta} + \sqrt{\frac{\omega_{\tau}(v)^2}{\beta ^2} + h_{\tau}(v,v)} \\
	&=& \pm \frac{\omega_{\tau}(v)}{\beta} + \sqrt{\rho_{\tau}(v,v)} , \quad \forall v\in T\Sigma.\nonumber
\end{eqnarray}

\begin{exe}[Conformastationary spacetimes revisited]\label{conformastationary}
{\rm Suppose $X \in \mathfrak{X}(M)$ is a complete timelike conformal Killing vector field on $(M,g)$, i.e., there exists $\sigma \in C^{\infty}(M)$ such that 
$$\mathcal{L}_Xg = \sigma \cdot g.$$
Javaloyes and S\'{a}nchez have shown \cite{java} if $(M,g)$ is a {\it distinguising} spacetime (i.e., for any $p,q\in M$, $I^{\pm}(p)\neq I^{\pm}(q)$ whenever $p\neq q$), then there exists a {\it spacelike} slice $\Sigma$ for the flow $\phi$ of $X$. Furthermore, there exist a 1-form $\omega_0 \in \Omega^1(\Sigma)$ and a Riemannian metric $h_0$ on $\Sigma$ and a positive smooth function $\Omega \in C^{\infty}(\mathbb{R}\times \Sigma)$ for which the pullback metric assumes the {\it standard conformastationary} form
\begin{equation}\label{conformaeq}
 \phi_{\Sigma}^{*}g = \Omega^2(\tau,x)\left( - d\tau^{2} + (\pi_2^{*}\omega_{0})\otimes d\tau + d\tau \otimes (\pi_2^{*}\omega_{0}) + \pi_2^{*}h_{0}\right). 
\end{equation}
 It is clear from (\ref{conformaeq}) that $\partial_{\tau}$ is a (future-directed unit timelike) {\it Killing } vector field for the conformally rescaled spacetime $(\mathbb{R}\times \Sigma, \Omega ^{-2}(\phi^{*}g))$. In this rescaled spacetime, 
 the causality of the spacetime is characterized in terms of the Finsler metric (see \cite{FHS13})
 \begin{equation}
 	F^{\pm}(v)= \pm\omega_0(v)+\sqrt{\omega_0(v)^2 +h_0(v,v)}, \quad \forall v\in T\Sigma.\label{eq:2}
 \end{equation}}
\end{exe}

With previous notations, we introduce first a convenient definition. 
\begin{defi}[Tame vector fields]\label{tamedefi}
The complete timelike vector field $X\in \mathfrak{X}(M)$ is {\it future tame} [resp. {\it past tame}] if its flow $\phi$ is a proper $\mathbb{R}$-action and for some slice $\Sigma \subset M$, there exist a positive function $\mathfrak{a}:\mathbb{R}\rightarrow \mathbb{R}$, a $1$-form $\omega_0$ and a Riemannian metric $g_0$ on $\Sigma$ such that $\forall \tau \in \mathbb{R}, \forall v\in T_x\Sigma,\forall x\in \Sigma,$
\begin{equation}\label{bound1}
	F^{+}_{\tau}(v) \leq \mathfrak{a}(\tau)\beta(\tau,x)F^{+}(v) \quad\mbox{ [ resp. $F^{-}_{\tau}(v) \geq \mathfrak{a}(\tau)\beta(\tau,x)F^{-}(v)$ ]. }
\end{equation}
where $F^{+}_{\tau}$ [resp. $F^{-}_{\tau}$] and $F^+$ [resp. $F^-$] are defined in \eqref{randerseq} and \eqref{eq:2}; and 
\begin{equation}\label{bound2}
	\int_{0}^{\infty}\frac{1}{\mathfrak{a}(s)}\, ds = \infty \quad\mbox{ [ resp. $\displaystyle\int_{-\infty}^{0}\dfrac{1}{\mathfrak{a}(s)}\, ds = \infty$].}
\end{equation}
If $X$ is both future and past tame, then we simply say it is {\it tame}. 
\end{defi}

\medskip

The next result uses future/past tameness as a key criterion for future/past omniscience. 
\begin{teo}\label{main1}
Let $X$ be a future [past] tame complete future-directed timelike vector on the spacetime $(M,g)$. Then the foliation defined by the integral curves of $X$ is future [past] omniscient. 
\end{teo}
\begin{proof}
The proofs for the future and past cases are analogous, so we just give the future case here. As above, we use the flow of $\phi$ to work on $(\mathbb{R}\times \Sigma, \phi_{\Sigma}^{*}g)$. Using the notation in (\ref{1formmetricinduced}), (\ref{eq:3}) and (\ref{randerseq}), observe that a curve $\sigma:[a,b]\rightarrow \mathbb{R}\times \Sigma$ with $\sigma(s)=(\tau(s),x(s))$ is a future-directed causal curve if and only if $\dot{\tau}(s)>0$ and:
\begin{equation}
\label{eq:5}
\beta(\tau(s),x(s))\dot{\tau}(s)\geq F^{+}_{\tau(s)}(\dot{x}(s)) \quad \forall s\in [a,b].
\end{equation}

\medskip

Fix $p \in M$ and an integral curve $\gamma$ of $X$ with $\phi_{\Sigma}(0,x_1) = \gamma(0) \in \Sigma$ for some $x_1 \in \Sigma$. Write $(\tau_0,x_0):= (\phi_{\Sigma})^{-1}(p)$. To establish omniscience, it is enough to show that one can find $\tau_1$ large enough so $(\tau_0,x_0)\ll (\tau_1,x_1)$ in $(\mathbb{R}\times \Sigma,\phi_{\Sigma}^{*}g)$, as this will imply that $p\ll \gamma(\tau_1)$ in $(M,g)$. 

Future tameness means that $\Sigma$ can be chosen so that there exist a positive function $\mathfrak{a}:\mathbb{R}\rightarrow \mathbb{R}$ satisfying \eqref{bound2}, a $1$-form $\omega_0$ and a Riemannian metric $g_0$ on $\Sigma$ such that 
\begin{equation}\label{boundproof}
    F^{+}_{\tau}(v) \leq \mathfrak{a}(\tau)\beta(\tau,x)F^+(v), \quad \forall \tau \in \mathbb{R}, \forall v\in T_x\Sigma,\forall x\in \Sigma,
\end{equation}
Since $\Sigma$ is connected we can pick a smooth curve $x:[0,\ell] \rightarrow \Sigma$ with $x(0) =x_0$ and $x(\ell) = x_1$ such that $F^+(\dot{x})=1$, i.e., $x$ is parametrized by $F^+$-arc length. Define $s:\mathbb{R}\rightarrow \mathbb{R}$ by 
\begin{equation}\label{boundproof2}
    s(t) := \int_{\tau_0}^t \frac{1}{\mathfrak{a}(\lambda)}\, d\lambda. 
\end{equation}
Let $\tau^{*}>\tau_0$ be such that $s(\tau^{*})=\ell$. Finally, define a curve  $\sigma(\lambda)=(\lambda,x(s(\lambda))[=:y(\lambda)])$ on $\mathbb{R}\times \Sigma$ with $\lambda \in [\tau_0,\tau^{*}]$. Clearly $\sigma(\tau_0)=(\tau_0,x_0)$ and $\sigma(\tau^{*}) = (\tau^{*},x_1)$. Moreover, using (\ref{boundproof}) we get 
\begin{eqnarray}
 F^{+}_{\lambda}(\dot{y}(\lambda)) &\leq& \mathfrak{a}(\lambda)\beta(\lambda,y(\lambda))F^+(\dot{y}(\lambda)) \nonumber \\
 &=& \beta(\lambda,y(\lambda)) \nonumber,
\end{eqnarray}
so $\sigma$ is future-directed causal by (\ref{eq:5}). That is, we have shown that $(\tau_0,x_0)\leq (\tau^{*},x_1) \ll (\tau_1,x_1)$, for any $\tau_1>\tau^{*}$. This completes the proof.
\end{proof}

\begin{rem}
{\rm Applying Theorem \ref{main1} in the context of Example \ref{conformastationary} we reestablish the Harris-Low result that the timelike vector field $X$ defines an omniscient foliation on $(M,-g(X,X)^{-1} g)$, and hence on the distinguishing conformastationary spacetime $(M,g)$.} 
\end{rem}

\section{Orthogonal splitting and omniscience}\label{sec2.2}

We consider here the particular case when the splitting provided in previous section is orthogonal, that is, $\omega_{\tau}\equiv 0$, and consequently, the metric on ${\mathbb R}\times \Sigma$ becomes 
\begin{equation}\label{ju}
g=-\beta^2\, d\tau^{2} + h_{\tau}\qquad\hbox{(abuse of notation here; compare with (\ref{eq:3})).}
\end{equation}
This simplification is specially relevant in the context of Bartnik conjecture because, as we will see later, it permits to take advantage of the TCC hypothesis in order to find better adapted conditions guaranteeing the omniscient character of the vector field $\partial_{\tau}$.

The main result of this section (Theorem \ref{io}) will require several hypotheses on both, the function $\beta^2$ and the spatial metric $h_{\tau}$. In order to gain a degree of freedom (useful to make these hypotheses more plausible), we will extract a generic conformal factor from the metric, and reparametrize the coordinate $\tau$ consistently, i.e. we will rewrite the metric (\ref{ju}) as follows 
\[
g=-\beta^2\, d\tau^{2} + h_{\tau}=\Omega^2\hat{g},\qquad\hbox{with $\hat{g}:=-dt^2+h_t$}.
\]
Clearly, this will be done at the price of probably loosing the complete character of the corresponding vector field 
$\partial_t$. Consequently, in the new coordinates we have $M\subset {\mathbb R}\times \Sigma$.

In this section the symbol $|\cdot|$ will be used to denote, indistinctly, the absolute value or the norm with metric $\hat{g}=-dt^2+h_t$. 
The symbol $|\cdot|_{t}$ will be reserved to denote the norm with metric $h_t$. Finally, the symbol $\|\cdot\|$ will denote the norm with metric $g$.

We begin with the following technical preliminar result.
\begin{lemma}\label{lll} Let $(M\subset {\mathbb R}\times \Sigma,\hat{g}=-dt^2+h_t)$ be a spacetime. Assume that
	\[
	|d^2/dt^2 h_t(\hat{v},\hat{w})|\quad\hbox{is bounded for all $h_0$-unit $\hat{v},\hat{w}\in T\Sigma$.}	
	\]
	Then, 
	\begin{equation}\label{xx}
		h_t(v,w)\leq \mathfrak{a}(t)^2h_0(v,w),\quad\hbox{with $\mathfrak{a}(t)=c_1t+c_2$,}\qquad\forall v,w\in T\Sigma.
	\end{equation}
	Moreover, given any finite interval $I$ with $I\times \Sigma\subset M$, the following terms, 
	\begin{equation}\label{ar'}
		\hat{{\rm Ric}}(\partial_t),\quad |H_{\Sigma_t}(x)|,\quad 
		d/dt\,|{\rm Jac}(Id_t)|,\quad 
		|\hat{\nabla}_{\partial_t}\partial_t|_t,\quad\hbox{are bounded above on $I\times \Sigma$,}
	\end{equation}
	with $H_{\Sigma_t}$ the mean curvature of $\Sigma_t\equiv (\Sigma,h_t)$ and $Id_t:\Sigma_t\rightarrow \Sigma_0$ the identity map.	
\end{lemma}
\begin{proof}
 From the hypothesis, the second derivative of metric coefficients for $h_t$ are bounded above on $M$. So, the first assertion is direct, and the last one follows from the fact that the expressions in local coordinates of the terms in (\ref{ar'}) only depend on metric coefficients and their (at most second order) derivatives. 
 \end{proof}

\smallskip

\begin{thm}\label{io} Let $(M\subset {\mathbb R}\times \Sigma,g=\Omega^2(-dt^2+h_t))$ be a future and past timelike geodesically complete and {\em cosmological} spacetime, i.e. it is globally hyperbolic with a compact Cauchy surface $\Sigma$ and satisfies the TCC. Suppose that ${\rm vol}(\Sigma_t)$, for all $t\in {\mathbb R}$, is bounded below by some positive number, and that
	\begin{equation}\label{tar'}
		|\nabla_{\Sigma_0}(\log\Omega)'(\cdot,\cdot)|_0,\quad 
		|d^2/dt^2 h_t(\hat{v},\hat{w})|\quad\hbox{are also bounded $\forall$ $h_0$-unit $\hat{v},\hat{w}\in T\Sigma$.}
	\end{equation}
	Then, $\partial_t$ is omniscient, and consequently, $(M,g)$ satisfies the NOH condition. 
\end{thm}

\begin{proof}
It suffices to show that $\partial_t$ is complete. In fact, in this case, we only need to show that $\partial_t$ is future omniscient for the conformally related spacetime $(M={\mathbb R}\times \Sigma,-dt^2+h_t)$ (for the past the argument is analogous). From Lemma \ref{lll} (recall (\ref{tar'})), we have 
\[
\sqrt{h_t(v,v)}\leq \mathfrak{a}(t)\sqrt{h_0(v,v)},\quad\hbox{with $\mathfrak{a}(t)=c_1t+c_2$,}\qquad\forall v\in T\Sigma,\; \forall t\in \mathbb{R},
\]
where, in particular,
\[
\int_0^\infty dt/\mathfrak{a}(t)=\infty.
\] 
Hence, $\partial_t$ is future tame in $({\mathbb R}\times \Sigma,-dt^2+h_t)$, and we conclude by applying Theorem \ref{main1}.


So, assume by contradiction that 
$\partial_t$ is incomplete in $M$. Let $0<t_{\infty}<\infty$ be such that $[0,t_{\infty})\times \Sigma\subset M$ is a non-precompact subset of $M$. 
Then, the function $\Omega$ must be
unbounded on $[0,t_{\infty})\times \Sigma\subset M$. In fact, assume by contradiction that $\Omega$ is bounded above by some positive number $C$. Let $\gamma:[0,\infty)\rightarrow M$, $\gamma(s)=(t(s),x(s))$, be a future-inextendible (and thus, future-complete) timelike geodesic with $t(s)\rightarrow t_{\infty}$ as $s\rightarrow\infty$. Then,
\[
\|\dot{\gamma}(s)\|=\sqrt{-g(\dot{\gamma}(s),\dot{\gamma}(s))}\leq\Omega(t(s),x(s))\dot{t}(s).
\]
Therefore,
\[
\infty={\rm length}(\gamma)=\int_0^{\infty}\|\dot{\gamma}(s)\|ds\leq \int_0^{\infty}\Omega(t(s),x(s))\dot{t}(s)ds\leq C\int_0^{\infty}\dot{t}(s)ds=C(t_{\infty}-t(0)),
\]
in contradiction to the fact that $t_{\infty}$ is finite. So, $\Omega$ is
unbounded on $[0,t_{\infty})\times \Sigma\subset M$.

The next ingredient is the well-known formula relating the Ricci tensor of the conformally related metrics $g$, $\hat{g}$ (see, e.g., p. 59 of \cite{besse})
\begin{eqnarray}
	\label{rescaling1'}
	Ric &=& \hat{Ric} + (1-n) \left[ \hat{Hess} _{\log \Omega} - d \log \Omega \otimes d \log \Omega \right] \nonumber \\
	& & -\left[ \hat{\triangle }\log \Omega + (n-1) \hat{g}(\hat{\nabla }\log \Omega, \hat{\nabla }\log \Omega) \right] \hat{g}.
\end{eqnarray}
If we evaluate formula (\ref{rescaling1'}) at $\partial_t\equiv \partial_t\mid_{(t,x)}$, we get, 
\begin{equation}\label{ki'}
	\begin{array}{c}
		Ric(\partial_t)= \hat{Ric}(\partial_t) + (1-n) \left[ \hat{Hess} _{\log \Omega}(\partial_t) - (\log\Omega)'(t,x)^2 \right] \\
		\qquad\qquad\qquad\quad +\left[ \hat{\triangle }\log \Omega(t,x) + (n-1) \hat{g}(\hat{\nabla }\log \Omega(t,x), \hat{\nabla }\log \Omega(t,x)) \right].
	\end{array}
\end{equation}
On the other hand, (recall \cite[Sublemma 14.33]{beembook}),
\[
\begin{array}{rl}
	\hat{\triangle }\log \Omega(t,x) & = \triangle_{\Sigma_t}\log \Omega(t,x)-\hat{g}(\hat{\nabla}\log\Omega(t,x),\partial_t)H_{\Sigma_t}(x)-\hat{{\rm Hess}}_{\log\Omega}(\partial_t) \\ & =\triangle_{\Sigma_t}\log \Omega(t,x)-(\log\Omega)'(t,x)H_{\Sigma_t}(x)-\hat{{\rm Hess}}_{\log\Omega}(\partial_t),
\end{array}
\]
where $\triangle_{\Sigma_t}$ and $H_{\Sigma_t}$ denote the laplacian and mean curvature associated to $\Sigma_t$, resp. So, formula (\ref{ki'}) can be rewritten as
\begin{equation}\label{kkii'}
	\begin{array}{c}
		Ric(\partial_t) =\hat{Ric}(\partial_t) -n \hat{Hess} _{\log \Omega}(\partial_t) + (n-1)(\log\Omega)'(t,x)^2 \\
		+\triangle_{\Sigma_t}\log \Omega(t,x)-(\log\Omega)'(t,x)H_{\Sigma_t}(x) + (n-1) \hat{g}(\hat{\nabla }\log \Omega(t,x), \hat{\nabla }\log \Omega(t,x)). 
	\end{array}
\end{equation}
Recall now that ($\nabla_{\Sigma_t}$ denotes the gradient in $\Sigma_t\equiv (\Sigma,h_t)$)
\begin{equation}\label{kii'}
	\begin{array}{rl}
		\hat{Hess}_{\log\Omega}(\partial_t) &
		=(\log\Omega)''-\hat{g}(\hat{\nabla}\log\Omega,\hat{\nabla}_{\partial_t}\partial_t) \\ & =(\log\Omega)''+\hat{g}((\log\Omega)'\partial_t,\hat{\nabla}_{\partial_t}\partial_t)-\hat{g}(\nabla_{\Sigma_t}\log\Omega,\hat{\nabla}_{\partial_t}\partial_t)
		\\ & =(\log\Omega)''+(\log\Omega)'\hat{g}(\partial_t,\hat{\nabla}_{\partial_t}\partial_t)-\hat{g}(\nabla_{\Sigma_t}\log\Omega,\hat{\nabla}_{\partial_t}\partial_t)
		\\ & =(\log\Omega)''-h_t(\nabla_{\Sigma_t}\log\Omega,\hat{\nabla}_{\partial_t}\partial_t),	
	\end{array}
\end{equation}
where we have used
\begin{equation}\label{kl'}
	\begin{array}{c}
		\hat{\nabla}\log\Omega(t,x)=-(\log\Omega)'(t,x)\partial_t+	
		\nabla_{\Sigma_t}\log\Omega(t,x)\quad\hbox{and}\quad \hat{g}(\partial_t,\hat{\nabla}_{\partial_t}\partial_t)\equiv 0.
	\end{array}
\end{equation}
Again from (\ref{kl'}) we have
\begin{equation}\label{kiii'}
	\hat{g}(\hat{\nabla }\log \Omega(t,x), \hat{\nabla }\log \Omega(t,x))=-(\log\Omega)'(t,x)^2+|\nabla_{\Sigma_t}\log\Omega(t,x)|_t^2.
\end{equation}
Therefore, if we take into account (\ref{kii'}) and (\ref{kiii'}) into (\ref{kkii'}), we get
\[
\begin{array}{c}
	n(\log\Omega)''(t,x)-\triangle_{\Sigma_t}\log \Omega(t,x)= -(\log\Omega)'(t,x)H_{\Sigma_t}(x) 
	+\hat{Ric}(\partial_t)-Ric(\partial_t) \qquad\qquad\qquad
	\\  \qquad\qquad\qquad\qquad\qquad\qquad\quad
	+(n-1)|\nabla_{\Sigma_t}\log\Omega(t,x)|_t^2
	+nh_t(\nabla_{\Sigma_t}\log\Omega,\hat{\nabla}_{\partial_t}\partial_t).	
\end{array}
\] 
Denoting $f(t,x):=\log\Omega(t,x)$, this last inequality can be rewritten as 
\begin{equation}\label{uno'}
	nf'' - \triangle_{\Sigma_t}f= -f'H_{\Sigma_t}(x) +C(t,x)+(n-1)|\nabla_{\Sigma_t}f|_t^2+nh_t(\nabla_{\Sigma_t}f,\hat{\nabla}_{\partial_t}\partial_t), 
\end{equation}
where the following terms are bounded above on $[0,t_{\infty})\times \Sigma$ by (\ref{tar'}) and Lemma \ref{lll}: 
\begin{equation}\label{ya'}
	|H_{\Sigma_t}|,\quad C(t,x):=\hat{Ric}(\partial_t)-Ric(\partial_t),\quad |\nabla_{\Sigma_t}f|_t^2,\quad h_t(\nabla_{\Sigma_t}f,\hat{\nabla}_{\partial_t}\partial_t). 
\end{equation}
Moreover, given $t_0\leq t<t_\infty$, the following expressions hold:
\[
\begin{array}{c}
	\int_{t_0}^t\int_{\Sigma_{\tau}}f'(\tau,x)f''(\tau,x)dxd\tau=	
	\\ =
	\frac{1}{2}\int_{t_0}^t\int_{\Sigma_{\tau}}\frac{d}{d\tau}f'(\tau,x)^2dxd\tau
	=\frac{1}{2}\int_{t_0}^t\int_{\Sigma_{0}}\frac{d}{d\tau}f'(\tau,x)^2|{\rm Jac}(Id_{\tau})|dxd\tau 
	\\ 
	=\frac{1}{2}\int_{t_0}^t\frac{d}{d\tau}\int_{\Sigma_{0}}f'(\tau,x)^2|{\rm Jac}(Id_{\tau})|dxd\tau
	-\frac{1}{2}\int_{t_0}^t\int_{\Sigma_{0}}f'(\tau,x)^2\frac{d}{d\tau}|{\rm Jac}(Id_{\tau})|dxd\tau
	\\ =
	\frac{1}{2}\int_{t_0}^t\frac{d}{d\tau}\int_{\Sigma_{\tau}}f'(\tau,x)^2dxd\tau
	-\frac{1}{2}\int_{t_0}^t\int_{\Sigma_{0}}f'(\tau,x)^2\frac{d}{d\tau}|{\rm Jac}(Id_{\tau})|dxd\tau
	\\ =
	\frac{1}{2}\int_{\Sigma_{t}}f'(t,x)^2dx-\frac{1}{2}\int_{\Sigma_{0}}f'(0,x)^2dx
	-\frac{1}{2}\int_{t_0}^t\int_{\Sigma_{0}}f'(\tau,x)^2\frac{d}{d\tau}|{\rm Jac}(Id_{\tau})|dxd\tau,
\end{array}
\]
(in the third and fifth equalities we have used the change of variables theorem applied to the diffeomorphism $Id_{\tau}:\Sigma_{\tau}\equiv (\Sigma,h_{\tau})\rightarrow \Sigma_{0}\equiv (\Sigma,h_{0})$). For $A>0$ big enough, we have
\[
\begin{array}{c}
	-\int_{t_0}^t\int_{\Sigma_{\tau}}f'(\tau,x)\Delta_{\Sigma_{\tau}}f(\tau,x)dxd\tau=\int_{t_0}^t\int_{\Sigma_{\tau}}h_{\tau}(\nabla_{\Sigma_{\tau}}f(\tau,x),\nabla_{\Sigma_{\tau}}f'(\tau,x))dxd\tau
	%
	%
	\\ \geq
	\frac{1}{2}\int_{t_0}^t\int_{\Sigma_{\tau}}\frac{d}{d\tau}|\nabla_{\Sigma_{\tau}}f(\tau,x)|_{\tau}^2dxd\tau  - A\int_{t_0}^t\int_{\Sigma_{0}}|\nabla_{\Sigma_0}f(\tau,x)|_0^2|{\rm Jac}(Id_{\tau})| dxd\tau
	\\ =
	\frac{1}{2}\int_{t_0}^t\int_{\Sigma_{0}}\frac{d}{d\tau}|\nabla_{\Sigma_{\tau}}f(\tau,x)|_{\tau}^2|{\rm Jac}(Id_{\tau})|dxd\tau  - A\int_{t_0}^t\int_{\Sigma_{0}}|\nabla_{\Sigma_0}f(\tau,x)|_0^2|{\rm Jac}(Id_{\tau})| dxd\tau
	\\  =
	\frac{1}{2}\int_{t_0}^t\frac{d}{d\tau}\int_{\Sigma_{0}}|\nabla_{S_{\tau}}f(\tau,x)|_{\tau}^2|{\rm Jac}(Id_{\tau})|dxd\tau 
	-\frac{1}{2}\int_{t_0}^t\int_{\Sigma_{0}}|\nabla_{\Sigma_{\tau}}f(\tau,x)|_{\tau}^2\frac{d}{d\tau}|{\rm Jac}(Id_{\tau})|dxd\tau 
	\\ 
	- A\int_{t_0}^t\int_{\Sigma_{0}}|\nabla_{\Sigma_0}f(\tau,x)|_0^2|{\rm Jac}(Id_{\tau})| dxd\tau
	\\ =
	\frac{1}{2}\int_{t_0}^t\frac{d}{d\tau}\int_{\Sigma_{\tau}}|\nabla_{\Sigma_{\tau}}f(\tau,x)|_{\tau}^2dxd\tau
	-\frac{1}{2}\int_{t_0}^t\int_{\Sigma_{0}}|\nabla_{\Sigma_{\tau}}f(\tau,x)|_{\tau}^2\frac{d}{d\tau}|{\rm Jac}(Id_{\tau})|dxd\tau \\ - A\int_{t_0}^t\int_{\Sigma_{0}}|\nabla_{\Sigma_0}f(\tau,x)|_0^2|{\rm Jac}(Id_{\tau})| dxd\tau
	\\ =
	\frac{1}{2}\int_{\Sigma_{t}}|\nabla_{\Sigma_{t}}f(t,x)|_t^2dx-\frac{1}{2}\int_{\Sigma_{0}}|\nabla_{\Sigma_{0}}f(0,x)|_0^2dx
	-\frac{1}{2}\int_{t_0}^t\int_{\Sigma_{0}}|\nabla_{\Sigma_{\tau}}f(\tau,x)|_{\tau}^2\frac{d}{d\tau}|{\rm Jac}(Id_{\tau})|dxd\tau \\ - A\int_{t_0}^t\int_{\Sigma_{0}}|\nabla_{\Sigma_0}f(\tau,x)|_0^2|{\rm Jac}(Id_{\tau})| dxd\tau,
\end{array}
\]
where, in the first equality, we have applied the identities
\[
{\rm div}_{\Sigma_{\tau}}(f'\nabla_{\Sigma_{\tau}}f)=h_{\tau}(\nabla_{\Sigma_{\tau}}f',\nabla_{\Sigma_{\tau}}f)+f'\Delta_{\Sigma_{\tau}}f,\qquad \int_{\Sigma_{\tau}}{\rm div}_{\Sigma_{\tau}}(f'\nabla_{\Sigma_{\tau}}f)dx=0, 
\]
and, for the inequality (in the second line), we have used that
	\[
	\begin{array}{c}
		h_{\tau}(\nabla_{\Sigma_{\tau}}f,\nabla_{\Sigma_{\tau}}f')=
		h_{\tau}^{ij}\frac{\partial f}{\partial x^i}\frac{\partial f'}{\partial x^j}=
		a^{ij}(\tau)h_0^{ij}\frac{\partial f}{\partial x^i}\frac{\partial f'}{\partial x^j} \\ =\frac{1}{2}\left(a^{ij}(\tau)h_0^{ij}\frac{\partial f}{\partial x^i}\frac{\partial f}{\partial x^j}\right)' - \frac{1}{2}{a^{ij}}'(\tau)h_0^{ij}\frac{\partial f}{\partial x^i}\frac{\partial f}{\partial x^j} \\
		\geq \frac{1}{2}\left(h_{\tau}^{ij}\frac{\partial f}{\partial x^i}\frac{\partial f}{\partial x^j}\right)' - A(\tau)h_0^{ij}\frac{\partial f}{\partial x^i}\frac{\partial f}{\partial x^j}=\frac{1}{2}\frac{d}{d\tau}|\nabla_{\Sigma_{\tau}} f|^2_{\tau}-A|\nabla_{\Sigma_0} f|^2_0\quad\hbox{for $A>0$ big enough.} 
	\end{array}
	\]
Next, if we multiply (\ref{uno'}) by $f'$, integrate it on $[t_0,t]$ and $\Sigma_t$ (recall that $\Sigma_t$ is compact), and use previous expressions,
we obtain that (\ref{uno'}) translates into the inequality,
\begin{equation}\label{equ1'}
	\begin{array}{rl}
		E_f(t)-E_f(t_0)\leq  &
		-\int_{t_0}^{t}\int_{\Sigma_{\tau}} f'(\tau,x)^2 H(\tau,x)dxd\tau + \int_{t_0}^{t}\int_{\Sigma_{\tau}} f'(\tau,x) C(\tau,x)dxd\tau \\ &
		+(n-1)\int_{t_0}^{t}\int_{\Sigma_{\tau}} f'(\tau,x) |\nabla_{\Sigma_{\tau}}f(\tau,x)|_{\tau}^2dxd\tau
		\\ &
		+n\int_{t_0}^{t}\int_{\Sigma_{\tau}} f'(\tau,x) h_t(\nabla_{\Sigma_{\tau}}f,\hat{\nabla}_{\partial_{\tau}}\partial_{\tau})dxd\tau
		\\ &
		+\frac{1}{2}\int_{t_0}^{t}\int_{\Sigma_{0}}(nf'(\tau,x)^2+|\nabla_{\Sigma_{\tau}}f(\tau,x)|_{\tau}^2)\frac{d}{d\tau}|{\rm Jac}(Id_{\tau})|dxd\tau \\ & + A\int_{t_0}^t\int_{\Sigma_{0}}|\nabla_{\Sigma_0}f(\tau,x)|_0^2|{\rm Jac}(Id_{\tau})| dxd\tau
	\end{array}
\end{equation}
where
\[
E_{f}(t):=\frac{1}{2}\int_{\Sigma_t} (nf'(t,x)^2 + |(\nabla_{\Sigma_t} f)(t,x)|_t^2)dx.
\]

In a first stage, assume that $f'(t,x)\geq 0$ for any $t_0\leq t<t_{\infty}$. If we take $\Lambda>0$ big enough,
the following inequality holds:
\begin{equation}\label{equ3'}
	\begin{array}{rl}
		\Lambda\int_{t_0}^{t}E_f(\tau)d\tau+\Lambda (t-t_0)\geq &
		-\int_{t_0}^{t}\int_{\Sigma_{\tau}} f'(\tau,x)^2 H(\tau,x)dxd\tau + \int_{t_0}^{t}\int_{\Sigma_{\tau}}f'(\tau,x) C(\tau,x)dxd\tau 
		\\ & +(n-1)\int_{t_0}^{t}\int_{\Sigma_{\tau}} f'(\tau,x)|\nabla_{\Sigma_{\tau}} f(\tau,x)|_{\tau}^2dxd\tau
		\\ &
		+n\int_{t_0}^{t}\int_{\Sigma_{\tau}} f'(\tau,x) h_t(\nabla_{\Sigma_{\tau}}f,\hat{\nabla}_{\partial_{\tau}}\partial_{\tau})dxd\tau
		\\ & +\frac{1}{2}\int_{t_0}^{t}\int_{\Sigma_{0}}(nf'(\tau,x)^2+|\nabla_{\Sigma_{\tau}}f(\tau,x)|_{\tau}^2)\frac{d}{d\tau}|{\rm Jac}(Id_{\tau})|dxd\tau \\ & +  A\int_{t_0}^t\int_{\Sigma_{0}}|\nabla_{\Sigma_0}f(\tau,x)|_0^2|{\rm Jac}(Id_{\tau})| dxd\tau.
	\end{array}
\end{equation}
Therefore, putting together (\ref{equ1'}) and (\ref{equ3'}), we have
\begin{equation}\label{contr'}
	E_{f}(t)-E_{f}(t_0)\leq
	\Lambda\int_{t_0}^{t}E_f(\tau)d\tau+\Lambda (t-t_0)\qquad\hbox{for all $t\in [t_0,t_{\infty})$.}
\end{equation}
By Grönwall's inequality,
\begin{equation}\label{ff''}
	E_f(t)\leq G(t)(<G(t_{\infty}))<\infty\qquad\hbox{for all $t\in [t_0,t_{\infty})$.}
\end{equation}
where
\[
G:[t_0,\infty)\rightarrow \mathbb{R},\qquad G(t):=(\Lambda(t-t_0)+E_f(t_0))e^{\Lambda (t-t_0)}.
\] 
Recall now that $f(t,x)$ and, consequently, $f'(t,x)$ are unbounded above on $[t_0,t_{\infty})\times S$; hence, there exists a sequence $\{(t_k,x_k)\}\subset M$ with $t_k\rightarrow t_{\infty}$ such that $f'(t_k,x_k)\rightarrow \infty$. This joined to the fact that
\begin{equation}\label{tarr'}
	|\nabla_{\Sigma_t}f'(t,x)|_t\leq k|\nabla_{\Sigma_0}f'(t,x)|_0\quad\hbox{is bounded on $M$ (recall (\ref{tar'})),}
\end{equation}
implies that $f'(t_k,x)\geq c_k\rightarrow\infty$ for any $x\in \Sigma$. In conclusion, 
\[
2E_f(t_k)\geq n\int_{\Sigma_{t_k}}f'(t_k,x)^2dx\geq n\, c_k^2\, {\rm vol}(\Sigma_{t_k})\rightarrow\infty, 
\]
in contradiction with the inequality (\ref{ff''}). 

Next, consider the remaining possibility, that is, $f'(t_k,x_k)<0$ for some sequence $\{(t_k,x_k)\}\subset M$ with $t_k\rightarrow t_{\infty}$. This joined to the fact that $f'$ is unbounded above on $[t_0,t_{\infty})\times \Sigma$, and the inequality (\ref{tarr'}), provides sequences $\{t^-_k\}, \{t^+_k\}\subset M$ with $t_k^-<t_k^+<t_{k+1}^-$ and $t_k^-, t_k^+\rightarrow t_\infty$ such that
\begin{equation}\label{cero'}
	f'(t_k^-,x)\leq c_0,\qquad f'(t_k^+,x)\geq c_k\rightarrow \infty\quad\forall x\in \Sigma
\end{equation}
\begin{equation}\label{unoll'}
	\
	f'(t,x)\geq 0\quad\forall x\in \Sigma,\quad\forall t\in [t_k^-,t_k^+]. 
\end{equation}
In fact, assume that $f'(t_k^+,x_k^+)\rightarrow\infty$. We can suppose without restriction that $t_k<t_k^+<t_{k+1}$. For each $k$, choose $t_k^-$ as the smallest value in $(t_k,t_k^+)$ satisfying $f'(t_k^-,x)\geq 0$ for all $x\in \Sigma$.

Now, we repeat the argument developed in previous case. More precisely, we begin with the inequality (deduced from (\ref{equ1'}) by taking into account (\ref{unoll'}))
\begin{equation}\label{contr''}
	E_{f}(t)-E_{f}(t_k^-)\leq
	\Lambda\int_{t_k^-}^{t}E_f(\tau)d\tau+\Lambda (t-t_k^-)\qquad\forall t\in [t_k^-,t_k^+]. 
\end{equation}
By Grönwall's inequality,
\begin{equation}\label{ff'}
	E_f(t)\leq G_k(t) \stackrel{(\ref{ya'})(\ref{cero'})}{<} (\Lambda t_{\infty}+nc_0^2D\,{\rm vol}(\Sigma_0)) e^{\Lambda t_{\infty}}\qquad\forall t\in [t_k^-,t_k^+],
\end{equation}
for some $D>0$, where
\[
G_k:[t_k^-,t_k^+]\rightarrow \mathbb{R},\qquad G_k(t):=(\Lambda(t-t_k^-)+E_f(t_k^-))e^{\Lambda (t-t_k^-)}.
\] 
This is in contradiction with the limit
\[ 
2E_f(t^+_k)\stackrel{(\ref{cero'})}{\geq} nc_k^2\, {\rm vol}(\Sigma_{t_k^+})\rightarrow\infty. 
\]
\end{proof}

\begin{rem} {\rm (1) From the proof of Theorem \ref{io} it becomes evident that TCC hypothesis can be weakened to the boundedness from below of ${\rm Ric}(\partial_t)$.
	
(2) It is well-known that any globally hyperbolic spacetime can be written in the form $({\mathbb R}\times \Sigma,f^2(-d\tau^2+h_{\tau}))$ (see \cite{bernalsanchez}). So, Theorem \ref{io} is applicable to any such spacetime admitting a reajustment of the  conformal factor compatible with its hypotheses. 
}
\end{rem}

\begin{exe} {\rm Let $({\mathbb R}\times \Sigma,g=-d\tau^2+\alpha^2(\tau,x) h_{0})$ be a timelike geodesically complete cosmological spacetime with $|(\nabla_{\Sigma_0}\alpha')(\tau,x)|_0$ bounded. The NOH condition can be guaranteed by applying Theorem \ref{io}. In fact, by making the change of coordinates $t=t(\tau,x)$ determined by the condition $\dot{t}(\tau,x)=\Omega^{-1}(\tau,x)$, with $\Omega:=\alpha$, we deduce  $$({\mathbb R}\times S,g=-d\tau^2+\alpha^2(\tau,x) h_{0})\cong (M\subset {\mathbb R}\times \Sigma,\alpha^2(t,x)(-dt^2+h_{0})),$$ which clearly falls under the hypotheses of Theorem \ref{io}. In particular, this shows that Bartnik conjecture holds for the class of warped spacetimes (i.e., $\alpha^2(\tau,x)\equiv \alpha^2(\tau)$).}
\end{exe}

\begin{exe} {\rm Let $({\mathbb R}\times \Sigma_1\times\cdots\times \Sigma_m,g=-d\tau^2+\sum_{i=1}^m\alpha_i^2(\tau) h_{i})$ be a multiwarped spacetime. Assume that it is timelike geodesically and cosmological. If there exists some function $\Omega^2:{\mathbb R}\rightarrow {\mathbb R}^+$ such that $d^2/dt^2(\alpha_i^2/\Omega^2)$ is bounded for all $i=1,\ldots,m$, then the spacetime satisfies the NOH condition. In fact, by making the change of coordinates $t=t(\tau)$ determined by the condition $\dot{t}(\tau)=\Omega^{-1}(\tau)$, we deduce  
	\[	
	\begin{array}{c}
		({\mathbb R}\times \Sigma_1\times\cdots\times \Sigma_m,-d\tau^2+\sum_{i=1}^m\alpha_i^2(\tau) h_{i})\cong\qquad\qquad\qquad\qquad\qquad\qquad \\ \qquad\qquad\qquad \cong (M\subset {\mathbb R}\times S_1\times\cdots\times S_m,\Omega^2(t)(-dt^2+\sum_{i=1}^m\alpha_i^2(t)/\Omega^2(t) h_{i})),
	\end{array} 
	\]
	which clearly falls under the hypotheses of Theorem \ref{io}. Therefore, Bartnik conjecture holds for this sub-class of multiwarped spacetimes. }  
\end{exe}

\section{Temporal functions and omniscience}\label{sec2.3}
In this section we wish to apply and illustrate the concepts of section \ref{sec2} to the important case of a {\it gradient} timelike vector field. However, we have seen those results require completeness of the underlying vector field; thus we must supplement our discussion with a digression on certain conditions which will ensure such completeness. 

\subsection{A digression: forms of completeness}\label{subsec1.1}
The Hopf-Rinow theorem for a Riemannian manifold $(N,h)$ states that the following conditions are equivalent.
\begin{itemize}
    \item[R1)] $(N,d_h)$ is a complete metric space, i.e., all Cauchy sequences converge. (Here, $d_h$ denotes the distance function associated with $h$.)
    \item[R2)] $(N,d_h)$ is {\it finitely compact}, that is, any $d_h$-bounded set has compact closure\footnote{Of course, this is also known as the Heine-Borel property, but we stick with the Busemann-Beem terminology here.}. Equivalently, any closed $d_h$-ball is compact. 
    \item[R3)] $(N,h)$ is geodesically complete.
\end{itemize}
It is well-known that there is no analogue of the Hopf-Rinow theorem for semi-Riemannian manifolds of indefinite signature, and in particular for Lorentzian manifolds. However, inspired directly by previous work by Busemann \cite{busemann} in the more abstract setting of the so-called {\em timelike spaces}, Beem introduced three separate conditions on spacetimes which, by analogy with the statements $(R1)-(R3)$ in the Hopf-Rinow theorem, might be construed as alternatives to the completeness of {\it causal geodesics} \cite{beem}, which in turn is important in physical applications. 

Let us briefly recall these concepts here (conf. also \cite{beem} and \cite[p. 211]{beembook}). 

\begin{defi}[Timelike Cauchy completeness]\label{tmlkcompdef}
A sequence $(x_k)_{k\in \mathbb{N}}$ on $M$ is {\em future [resp. past] timelike Cauchy} if 
\begin{itemize}
    \item[i)] $x_k\ll x_{k+1}$ [resp. $x_{k+1}\ll x_k$], $\forall k \in \mathbb{N}$;
    \item[ii)] there exists a sequence $(B_k)_{k\in \mathbb{N}}$ of nonnegative real numbers such that 
    $$d(x_k,x_{k'})\leq B_k \quad\mbox{ [resp. $d(x_{k'},x_k) \leq B_k$]}, \quad \forall k\leq k'$$
    and with $B_k \rightarrow 0$. We call $(B_k)$ a {\em bounding sequence}.
\end{itemize}
We say that $(M,g)$ is {\em future [resp. past] timelike Cauchy complete} if any future [resp. past] timelike Cauchy sequence converges. If $(M,g)$ is both future and past timelike Cauchy complete, then we simply say it is {\em timelike Cauchy complete}. 
\end{defi}

\begin{defi}[Finite compactness]\label{fincompdef}
$(M,g)$ is said to be {\em future [resp. past] finitely compact} if for any $p,q \in M$ and $B \in \mathbb{R}$, 
$$p\ll q \Rightarrow K=\{x\in J^{+}(q) \; : \;  d(p,x)\leq B\} \mbox{ is compact}$$
$$\mbox{[resp. }q\ll p \Rightarrow K=\{x\in J^{-}(q) \; : \; d(x,p)\leq B\} \mbox{ is compact]}.$$
$(M,g)$ is {\em finitely compact} if it is both future and past finitely compact.
\end{defi}
\begin{defi}[condition $A$]\label{condA}
$(M,g)$ is said to satisfy the {\em future [resp. past] condition $A$} if for any $p, q\in M$ with $p\ll q$ [resp. $q\ll p$] and any future-[resp. past-]inextendible causal geodesic $\gamma:[0,b) \rightarrow M$ ($0<b\leq +\infty$) with $\gamma(0)=q$ we have $d(p,\gamma(t)) \rightarrow +\infty$ [resp. $d(\gamma(t),p) \rightarrow +\infty$] as $t\rightarrow b$. 
$(M,g)$ satisfies the {\em condition $A$} if it satisfies both future and past conditions $A$. 
\end{defi}

Even if $(M,g)$ is globally hyperbolic these conditions do {\em not} imply timelike geodesic completeness (conf., e.g., the example due to Geroch illustrated in \cite[Fig. 6.2]{beembook}). However, we have

\begin{prop}\label{condAcomplete}
Suppose $(M,g)$ satisfies the future condition $A$. Then, any future-directed timelike geodesic ray is future-complete. An analogous statement holds for {\em past-directed} timelike geodesic rays if $(M,g)$ satisfies the past condition $A$. In particular, if $(M,g)$ satisfies condition $A$, then any timelike geodesic line is complete. 
\end{prop}
\noindent {\em Proof.} We prove the future case only, as the past follows by time-duality. Assume then that $(M,g)$ satisfies the future condition $A$, and let $\gamma:[0,b)\rightarrow M$ be a future-directed timelike geodesic ray. By definition, this means that $\gamma$ is future-inextendible and $L_g(\gamma |_{[0,t)}) = d(\gamma(0),\gamma(t))$ for each $t\in (0,b)$. We wish to prove that $b=+\infty$. Take $p:=\gamma(0)$ and $q:=\gamma(a)$ for some $0<a<b$. Now, $\tilde{\gamma}: s\in [0,b-a) \mapsto \gamma(s+a) \in M$ is still a timelike geodesic ray issuing from $q$, and $p\ll q$. Thus 
$$d(p,\tilde{\gamma}(s)) = L_g(\gamma|_{[0,s+a)})= |\dot{\gamma}(0)|(s+a).$$
Since $d(p,\tilde{\gamma}(s))\rightarrow +\infty$ as $s\rightarrow b-a$ due to the future condition $A$, and the right-hand side of the previous equation diverges, we conclude that $b=\infty$ as desired. The last statement follows by applying the future and past cases to the two ``halves'' of a timelike geodesic line. 
\qcd

Finite compactness, timelike Cauchy completeness and the condition $A$ are not independent if $(M,g)$ has sufficiently good causality. 
\begin{prop}\label{equivprop}
Suppose $(M,g)$ is strongly causal. Consider the following statements.
\begin{itemize}
    \item[a)] $(M,g)$ is future finitely compact.
    \item[b)] $(M,g)$ is future Cauchy complete.
    \item[c)] $(M,g)$ satisfies the future condition $A$.
\end{itemize}
Then, $$(a)\Rightarrow (b) \Rightarrow (c).$$
If $(M,g)$ is globally hyperbolic, then $(c) \Rightarrow (a)$, that is, all the  statements $(a)-(c)$ are equivalent. Analogous statements hold for the past versions of all these statements.
\end{prop}
\noindent {\it Proof.} Again we only need to show the future case.\\
$(a) \Rightarrow (b)$\\
Let $x=(x_k)_{k\in \mathbb{N}}$ be future timelike Cauchy sequence in $M$ with an associated bounding sequence $(B_k)_{k \in \mathbb{N}}$. Observe that $x_1\ll x_2$, so
$$x_k \in K= \{z \in J^{+}(x_2) \, : \, d(x_1,z) \leq B_1\}, \quad \forall k\geq 2.$$
Since $K$ is compact by assumption, we can assume that some subsequence $(x_{k_i})_{i\in \mathbb{N}}$ of $x$ converges to some point $z_0 \in K$. Let $U\ni z_0$ be any open set. Using strong causality we can pick a causally convex neighborhood $V\subset U$ of $z_0$. Let $i_0 \in \mathbb{N}$ such that $x_{k_i} \in V$ whenever $i\geq i_0$. Pick any $k> k_{i_0}$. Since for large enough $i>i_0$ we have $k_i>k$, it follows that 
$$x_{k_{i_0}} \ll x_k \ll x_{k_i},$$
and hence $x_k \in V$ by causal convexity. We conclude that $(x_k)$ itself converges to $z_0$. \\
$(b)\Rightarrow (c)$\\
We shall slightly modify the proof of \cite[Lemma 4]{beem}, which assumes global hyperbolicity. Let $p\ll q$ and a future-inextendible causal geodesic $\gamma:[0,b)\rightarrow M$ starting at $q$. Suppose the condition 
$$\lim_{t\rightarrow b}d(p,\gamma(t)) =+\infty$$
{\em fails}. In that case, there exists some constant $B$ and a sequence $(t_k)_{k\in \mathbb{N}}$ in $[0,b)$ converging to $b$ for which 
\begin{equation}\label{eq1}
   d(p,\gamma(t_k)) \leq B, \quad \forall k \in \mathbb{N}.
\end{equation}
We shall presently construct a future timelike Cauchy sequence which does not converge, so that $(M,g)$ is not future timelike Cauchy complete. Fix an auxiliary Riemannian metric $h$ on $M$ with distance function $\rho_h$. Since 
$\gamma(t_k) \in I^{+}(p)$ and $\gamma(t_k)\leq \gamma(t_{k'})$ for $k<k'$, it is not difficult to construct inductively a sequence $(z_k)_{k\in \mathbb{N}}$ in $I^{+}(p)$ for which
\begin{eqnarray}
z_k &\ll& \gamma(t_k) \label{try1} \\
z_{k}&\ll & z_{k+1} \label{try2} \\
\rho_h(z_k,\gamma(t_k))&<&1/k, \quad \forall k \in \mathbb{N}. \label{try3}
\end{eqnarray}
Thus, using the reverse triangle inequality, first together with (\ref{try1}) and (\ref{eq1}), we get
$$d(p_z,k) \leq d(p,z_k) + d(z_k,\gamma(t_k)) \leq d(p, \gamma(t_k))\leq B,$$
and together with (\ref{try2}), 
\begin{equation}\label{eq2}d(p,z_k) \leq d(p,z_k) + d(z_k,z_{k+l}) \leq d(p,z_{k+l}),\qquad\forall l,k\in\mathbb{N}.\end{equation}
We conclude that $d(p,z_k)$ forms a bounded increasing sequence of real numbers, and finally, (\ref{eq2}) can be rewritten as 
$$d(z_{k},z_{k+l})\leq d(p,z_{k+l})-d(p,z_k),$$
which establishes that the sequence $(z_k)_{k\in \mathbb{N}}$ is future timelike Cauchy. We claim it does not converge. Indeed, if it did, then (\ref{try3}) would imply that so would $(\gamma(t_k))_{k\in \mathbb{N}}$. However, this is not possible because it would mean that the future-inextendible causal curve $\gamma$ would be partially imprisoned in a compact set, which in turn is prohibited by strong causality. 

Finally. if $(M,g)$ is globally hyperbolic, the proof that $(c)\Rightarrow (a)$ is given in \cite[Theorem 5]{beem}, so we do not repeat it here.
\qcd

We end this digression by stating Corollary 6 of \cite{beem}, which gives a nice sufficient condition for the completeness conditions introduced here. 
\begin{teo} \label{complete}
If $(M,g)$ is globally hyperbolic and causally geodesically complete, then $(M,g)$ is finitely compact, and hence it is Cauchy complete and satisfies condition $A$. 
\end{teo}
\qcd

\subsection{The timelike eikonal equation}

Returning to our main discussion in this section, recall that a smooth real-valued function $f\in C^{\infty}(M)$ is said to obey (or be a solution of) the {\em timelike eikonal inequality} if $\langle \nabla f,\nabla f \rangle <0$. Such a function is also called a {\em temporal function} in the literature. It is a basic fact in causality theory that a necessary and sufficient condition for the existence of a smooth global solution to the eikonal inequality is that $(M,g)$ is {\em stably causal} (see, e.g., \cite[Proposition 6.4.9]{HE} and \cite[Theorem 1.2]{bernalsanchez}). 

A temporal function $f$ is said to satisfy the {\em timelike eikonal equation} if $$\langle \nabla f,\nabla f\rangle =-1.$$ 
Now, the existence of a (global smooth) solution to the eikonal equation is, at least in principle, a much stronger requirement on a spacetime than just stable causality. In gravitational physics applications, it describes a so-called {\em proper-time synchronizable observer field} \cite[p. 358]{oneill}. It also has a rich geometric content; we list a few of its general properties which we will use here (see the Appendix B of \cite{beembook} and references therein for more details and proofs). Let $f:M\rightarrow \mathbb{R}$ be such a solution to the eikonal equation. 
\begin{itemize}
    \item[E1)] Since $f$ has everywhere timelike gradient, any (non-empty) level set $f^{-1}(r)$ is a smooth spacelike (embedded) hypersurface whose second fundamental form is given (up to a sign choice) by the Hessian $H_f$ of $f$ computed on tangent vectors $v\in Tf^{-1}(r)$. Level hypersurfaces are also acausal in $(M,g)$, because given any (piecewise smooth) causal curve $\alpha:[a,b] \rightarrow M$ we have 
\begin{equation}\label{level1}
(f\circ \alpha)(b)-(f\circ \alpha)(a) = \int_a^b(f\circ \alpha)'(t)\, dt = \int_a^b \langle \nabla f(\alpha(t)),\dot{\alpha}(t)\rangle \, dt \neq 0. 
\end{equation}
In particular, if $\alpha$ is an integral curve of $\nabla f$, i.e., $\dot{\alpha}= \nabla f\circ \alpha$, then (\ref{level1}) becomes
\begin{equation}\label{level1.5}
   (f\circ \alpha)(b)-(f\circ \alpha)(a) = a-b; 
\end{equation}
Therefore, {\it if $\nabla f$ is complete, then $f(\gamma(\mathbb{R}))=\mathbb{R}$ for any inextendible integral curve $\gamma$}, a fact we shall need later on. 
\item[E2)] From (\ref{level1}) we easily deduce that the Lorentzian length $L_g(\alpha)$ of a causal curve segment $\alpha$ satisfies the inequality
\begin{equation}\label{level3}
  L_g(\alpha) \leq |f(\alpha(b)) - f(\alpha(a))|,
\end{equation}
with equality if and only if $\alpha$ is a reparametrization of a segment of integral curve of $\nabla f$. In particular, {\em any segment of integral curve of $\nabla f$ with endpoints on level sets $f^{-1}(s)$ and $f^{-1}(t)$ with $s\neq t$ has maximum Lorentzian length ($\equiv |s-t|$, cf. (\ref{level1.5})) among all causal curve segments connecting $f^{-1}(s)$ and $f^{-1}(t)$.} 
\item[E3)] $\nabla _{\nabla f}\nabla f= 0$. Thus, any integral curve of $\nabla f$ is a unit timelike geodesic. Moreover the previous item means these timelike geodesics are maximal, that is, they maximize the Lorentzian distance function between any two of its points.
\item[E4)] If we consider the decomposition of a vector field $X\in \mathfrak{X}(M)$ into parts parallel and orthogonal to $\nabla f$:
\begin{equation}\label{ortdecomp}
  X= - \langle X, \nabla f\rangle \cdot \nabla f + X^{\perp}.   
\end{equation}
Using this decomposition we may define the following $(0,2)$-tensor, which we refer to as the {\em spatial part} of the metric $g$: 
\begin{equation}\label{normmetric}
h(X,Y) := g(X^{\perp},Y^{\perp}).
\end{equation}
Thus, one easily checks that (\ref{ortdecomp}) implies that the spacetime metric $g$ decomposes as
\begin{equation}\label{metricdecomp}
  g = -df \otimes df + h.
\end{equation}
Observe that the spatial part of the metric is positive-definite on $\nabla f^{\perp}$, and it immediately follows from (\ref{metricdecomp}) that it coincides at every point with the (Riemannian) induced metric on the level hypersuface  through that point.
\end{itemize}

\medskip 

The properties just listed of solutions to the timelike eikonal equation are fairly ostensible. We shall also need some slightly more specialized ones, which we list in the following

\begin{prop}\label{propeiko}
Let $f \in C^{\infty}(M)$ be a solution to the timelike eikonal equation such that $\nabla f$ is complete, and let $\Sigma := f^{-1}(0)$. The following statements hold.
\begin{itemize}
    \item[i)] Let $\phi: \mathbb{R}\times M \rightarrow M$ denote the flow of $X=\nabla f$, and let $\phi_{\Sigma}:= \phi|_{\mathbb{R}\times \Sigma}$ (recall \eqref{eq:1}). Then $\phi_{\Sigma}$ is a diffeomorphism between $\mathbb{R}\times \Sigma$ and $M$, and in fact 
    \begin{equation}\label{normexp}
        \phi_{\Sigma}(t,x) = \exp^{\Sigma} _{\perp}(t\cdot \nabla f(x)), \quad \forall t \in \mathbb{R},\;\; \forall x\in \Sigma,
    \end{equation}
    where $\exp^{\Sigma} _{\perp}$ denotes the normal exponential map of $\Sigma $. In particular, $\Sigma$ is a slice for $\phi$, so the action is proper.
    \item[ii)] Let $x \in \Sigma$ and $u\in T_x\Sigma$. Then $J:\mathbb{R}\rightarrow TM$ given by
    \begin{equation}\label{jacobi}
    J_u(t) := d(\phi_\Sigma)_{(t,x)}(0,u), \quad \forall t\in \mathbb{R}
    \end{equation}
    is a Jacobi field along the integral curve of $\nabla f$ through $x$. Indeed, it is the Jacobi field associated with a geodesic variation 
    $$\sigma(t,s) = \exp^{\Sigma} _{\perp}(t\cdot \nabla f(\beta(s)),$$
    where $\beta: (-\delta,\delta)\rightarrow \Sigma$ is any smooth curve with $\dot{\beta}(0)=u$. Furthermore, $J_u$ satisfies the equation 
    \begin{equation}
        \dot{J}_u(t) = -L_t(J_u(t)),
    \end{equation}
    $J_u(0)=u$ and $\dot{J}_u(0) = -L_0(u)$, where $L_t(v):= -\nabla _v\nabla f$ is the Weingarten shape operator of the level hypersurface $f^{-1}(-t)$. (In particular, $J_u$ is a $\Sigma$-Jacobi field). 
\end{itemize}
\end{prop}
\noindent {\it Comments on the proof.} $(i)$ follows immediately from the fact that integral curves of $\nabla f$ are complete timelike geodesic lines (conf. properties $(E1)-(E3)$ above), which means, in particular, that the normal exponential map has no singularities (which would correspond to focal points of $\Sigma$ along the integral curves). $(ii)$ is proved in detail in the section 2 of the Appendix B of \cite{beembook} (see especially Eq. (B.12) and Lemma B.6 of that reference).
\qcd

\subsection{Main result}

Now, we proceed to find the tight condition that guarantes the onmiscience property in our class of spacetimes. To this aim, first let us consider ($2$-dimensional) de Sitter's spacetime $$(\mathbb{R}\times\mathbb{S}^1,-dt^2+\mathfrak{a}_k(t)^2d\Omega_0),\qquad \mathfrak{a}_k(t):=\cosh(kt).$$ 
This spacetime becomes relevant because the corresponding gradient vector field $\nabla f=\partial_t$ lies just in the limit of not verifying the omniscience property. In fact, in this case, the function $\mathfrak{a}_k$ asymptotically grows to infinity too fast, verifying
\[
\int_0^{\infty}\frac{dt}{\mathfrak{a}_k(t)}<\infty;\qquad\hbox{moreover,}\quad -\frac{\ddot{\mathfrak{a}}_k}{\mathfrak{a}_k}=k^2(\equiv\hbox{curvature}).
\]  
This suggests that, in order to ensure the omniscient character of $\nabla f$, it suffices to bound the curvature by some quotient of the form $-\ddot{\mathfrak{a}}/\mathfrak{a}$, for some positive function $\mathfrak{a}$ with moderate asymptotic growth, in the sense that the corresponding integral $\int_0^{\infty}\mathfrak{a}^{-1}$ is infinite. This simple observation allows us to establish the following key technical result.

%
%
\begin{lemma}\label{tamee}
	Let $(M=\mathbb{R}\times \Sigma,g=-dt^2+h_t)$ be a spacetime such that the integral curves of $\partial_t$ are geodesics.
	Suppose the existence of some positive function $\mathfrak{a}:[0,\infty)\rightarrow \mathbb{R}$, with $\dot{\mathfrak{a}}(t_0)>0$ for some $t_0\in [0,\infty)$, 
	such that
	\begin{equation}\label{hypp}
		\sqrt{n}R_t\geq -\frac{\ddot{\mathfrak{a}}(t)}{\mathfrak{a}(t)}\quad\forall t\in [t_0,\infty)\qquad\hbox{and}\qquad \int_{0}^{\infty}\frac{dt}{\mathfrak{a}(t)}=\infty,
	\end{equation}
	where $R_t:={\rm max}\{\langle R(\partial_t,\hat{v})\partial_t,\hat{w}\rangle :\; \hat{v},\hat{w}\in \partial_t^{\perp}\;\hbox{are $g$-unit}\}$, and $R$ is the curvature tensor. Then, there exists $\epsilon>0$ small enough such that 
	\begin{equation}\label{tameq}
		\epsilon \sqrt{h_t(J_u(t),J_u(t))} \leq \mathfrak{a}(t) \cdot \sqrt{h_0(u,u)}, \quad \forall u\in T\Sigma,\; \forall t\in [0,\infty),
	\end{equation}
	where $J_u$ is given as in Eq. (\ref{jacobi}).
\end{lemma}
\begin{proof}
Let $\{E_i(t)\}_{i=1}^n$ be an orthonormal frame of parallel vector fields (all of them orthogonal to $\partial_t$) along a generic integral curve of $\partial_t$. Then
\[
J_u(t)=\sum_{i=1}^n b_i(t)E_i(t),\quad \ddot{J}_u(t)=\sum_{i=1}^n \ddot{b}_i(t)E_i(t),\quad |J_u(t)|=\sqrt{\sum_{i=1}^n b_i(t)^2}.
\]
Since $J_u$ is a Jacobi vector field along the integral curve of $\partial_t$, it satisfies:
\begin{equation}\label{oio}
	\ddot{J}_u+R(\partial_t,J_u)\partial_t=0,\qquad J_u(0)=u,\qquad\dot{J}_u(0)=-L(u),
\end{equation}
being $L$ the Weingarten shape operator of $\{0\}\times\Sigma$ in $(M,g)$. Therefore, if we scalarly multiply the equation in (\ref{oio}) by $E_i(t)$, we obtain
\[
\langle \ddot{J}_u(t),E_i(t)\rangle =
-\langle R(\partial_t,J_u(t)|J_u(t)|^{-1})\partial_t,E_i(t)\rangle |J_u(t)|;
\]
denoting $R_{u,i}(t):=\langle R(\partial_t,J_u(t)|J_u(t)|^{-1}),\partial_t,E_i(t)\rangle$, it can be rewritten as
\[
\ddot{b}_i(t)=-R_{u,i}(t)\sqrt{b_1(t)^2+\cdots+b_n(t)^2}\qquad \forall i=1,\ldots,n,\qquad\forall t\in \mathbb{R}.
\]
But, taking into account the hypothesis (\ref{hyp}), we also have
\[
\ddot{a}(t)\geq -\sqrt{n}R_{u,i}(t)a(t)\qquad\forall i=1,\ldots,n\qquad\forall t\in [0,\infty).
\]
So, a standard comparison argument allows us to conclude the existence of $\epsilon>0$ small enough such that
\[
\epsilon
\sqrt{ h(J_u(t),J_u(t))}=\epsilon\sqrt{b_1(t)^2+\cdots+b_n(t)^2}<\mathfrak{a}(t)=\mathfrak{a}(t)\sqrt{h_0(u,u)}\quad\forall t\in [0,\infty).
\]
In fact, assume by contradiction that the inequality in previous expression does not hold. Let $t_*\in (0,\infty)$ be the infimum value such that $\epsilon b_{i}(t)<\mathfrak{a}(t)$ for all $i$ and for all $t\in [0,t_*)$ and $\epsilon b_{i_0}(t_*)=\mathfrak{a}(t_*)$ for some $i_0$. Then, $\epsilon \dot{b}_{i_0}(t_*)\geq \dot{\mathfrak{a}}(t_*)$, and thus, $\epsilon \ddot{b}_{i_0}(t_{**})>\ddot{\mathfrak{a}}(t_{**})$ for some $0<t_{**}<t_*$, in contradiction with
	\[
	\epsilon\ddot{b}_{i_0}(t_{**})=-\epsilon k_{i_0}(t_{**})\sqrt{\sum_{i=1}^{n}b_{i}(t_{**})^2}<-\sqrt{n}k_{i_0}(t_{**})\mathfrak{a}(t_{**})\leq\ddot{\mathfrak{a}}(t_{**}).
	\]
\end{proof}

\begin{lemma}\label{tame}
Let $f$ be a solution of the timelike eikonal equation, and assume $\nabla f$ is complete. Suppose, in addition, that the curvature constraint (\ref{hypp}) holds (with $\nabla f\equiv \partial_t$).
Then $\nabla f$ is future tame. 
\end{lemma}
\begin{proof}
First, observe that the completeness of $\nabla f$ means that $f(M) =\mathbb{R}$ and $\Sigma :=f^{-1}(0)$ is a slice for the flow $\phi$ of $\nabla f$ by Proposition \ref{propeiko}. Referring to the notation in section \ref{sec2}, observe that (\ref{metricdecomp}) yields 
\begin{equation}\label{metricdecompsplit}
\phi_{\Sigma}^{*}g = -d\tau^2 + h_{\tau},
\end{equation}
where $h_{\tau} = \phi_{\Sigma}^{*}h$ is a 1-parameter family of metrics on $\Sigma$. Using this and (\ref{metricdecompsplit}) we find out the associated family of pre-Finsler metrics reduces to 
\begin{equation}\label{notfinsler}F^{\pm}_{\tau} \equiv \sqrt{h_{\tau}}.\end{equation}
But then, applying Lemma \ref{tamee}, and using Eq. (\ref{jacobi}) in (\ref{tameq}), we conclude that it becomes
$$\sqrt{h_{\tau}(v,v)} \leq \mathfrak{a}(\tau)\epsilon^{-1} \sqrt{h_{0}(v,v)}, \quad \forall v\in T\Sigma,$$
which in view of (\ref{notfinsler}) precisely gives the condition of future tameness. (Compare Definition \ref{tamedefi}.)
\end{proof}
As mentioned above, global solutions to the eikonal equation are hard to come by; however, in the 1990s Garc\'{i}a-R\'{i}o and Kupeli \cite{garciario1,garciario2} made a simple and yet fruitful observation: if $f\in C^{\infty}(M)$ is a {\em temporal} function, i.e., if it a priori satisfies only the eikonal {\em inequality} - which, recall, exists in any stably causal spacetime - then, {\em with respect to the conformally related metric}
\begin{equation}\label{grkmetric}
  \hat{g} := -g(\nabla f,\nabla f)\cdot g,  
\end{equation}
$f$ is a solution of the eikonal {\em equation}, i.e., $|\hat{\nabla} f|_{\hat{g}}=1$ on $(M,\hat{g})$. They used this fact to obtain a number of splitting statements for stably causal spacetimes. This simple fact will be an important ingredient in our main result in this section, but we shall need a technical lemma first. 
\begin{lemma}\label{lemma1}
Suppose $(M,g)$ is future timelike Cauchy complete [resp. future finitely compact], and let $\Omega \in C^{\infty}(M)$ be a strictly positive function such that $\Omega \geq \varepsilon$ for some number $\varepsilon >0$. Then, $(M,\hat{g})$ is also future timelike Cauchy complete [resp. future finitely compact], where
$$\hat{g}:= \Omega ^2g.$$
(Analogous results hold for the past cases.) 
\end{lemma}
\begin{proof} Let $p\leq_{g} p$, and let $\alpha:[a,b]\rightarrow M$ be a future-directed causal curve segment in $(M,g)$ from $p$ to $q$. Thus it is also a future-directed causal curve segment in $(M,\hat{g})$ (so $p\leq _{\hat{g}}q$). Then we have 
$$L_g(\alpha) = \int_{a}^b (\Omega(\alpha(t))^{-1}|\dot{\alpha}(t)|_{\hat{g}}\, dt \leq (1/\varepsilon) L_{\hat{g}}(\alpha),$$
whence we conclude that 
\begin{equation}\label{keyineq}
d(p,q) \leq (1/\varepsilon) \hat{d}(p,q),
\end{equation}
where $\hat{d}$ is the Lorentzian distance function with respect to $\hat{g}$. Thus, we have:
\begin{itemize}
    \item[1)] If $(x_k)_{k\in \mathbb{N}} $ is a future timelike Cauchy sequence in $(M,\hat{g})$ with bounding sequence $(B_k)$, then (\ref{keyineq}) implies that $(x_k)_{k\in \mathbb{N}} $ is also a future timelike Cauchy sequence in $(M,g)$ with bounding sequence $(B_k/\varepsilon)$. This implies that if $(M,g)$ is future timelike Cauchy complete, then so is $(M,\hat{g})$.
    \item[2)] Assume now $(M,g)$ is future finitely compact. Let $p,q \in M$ and $B \in \mathbb{R}$. If $p\ll _{\hat{g}} q$, then also $p \ll _g q$, so 
$$\hat{K}:= \{x\in J_{\hat{g}}^{+}(q) \; : \;  \hat{d}(p,x)\leq B\} \stackrel{(\ref{keyineq})}{\subset} \{x\in J_g^{+}(q) \; : \;  d(p,x)\leq B/\varepsilon\} := K,$$
and $K$ is compact. By the lower semicontinuity of the Lorentzian distance function of $(M,g)$, and the fact that $K$ is closed implies that $\hat{K}$ is closed, and hence compact. Thus, $(M, \hat{g})$ is indeed future finitely compact. 
\end{itemize}

\end{proof}
We are now ready to state and prove the main result of this section.  
\begin{teo}\label{mainthm}
Suppose $(M,g)$ is timelike Cauchy complete and admits a temporal function $f\in C^{\infty}(M)$ such that 
\begin{itemize}
    \item[a)] there exists a number $\varepsilon>0$ for which $| \nabla f|_g \geq \varepsilon$;
    \item[b)] $f$ is a solution of the eikonal equation on $(M,\hat{g}:=-g(\nabla f,\nabla f)g)$ satisfying the curvature constraint (\ref{hypp}) with $\partial_t$ replaced by $\nabla f$.
\end{itemize}
Then $f(M) = \mathbb{R}$ and the timelike foliation spanned by $\nabla f$ is future omniscient in $(M,g)$. If in addition $f^{-1}(0)$ is compact, then the NOH condition holds in $(M,g)$.
\end{teo}
\begin{proof} First, recall that future omniscience is a conformally invariant condition, so it suffices to show it on $(M,\hat{g})$. Moreover, a) together with Lemma \ref{lemma1} imply that $(M,\hat{g})$ is itself timelike Cauchy complete. Therefore, we may simply drop the hat from the metric and assume $f$ to be a solution to the eikonal equation on $(M,g)$ itself, satisfying the curvature constraint (\ref{hypp}), 
and we do so for the rest of the proof. 

Next, observe that the stable causality of $(M,g)$ implies it is also strongly causal, so condition $A$ holds by Proposition \ref{equivprop}, and Proposition \ref{condAcomplete} now implies that $\nabla f$ is complete. Moreover, Lemma \ref{tame} ensures that $\nabla f$ is future tame, and so the conclusion follows using Theorem \ref{main1}. The last statement then follows from Proposition \ref{omnitonoh}.
\end{proof}

\begin{cor}\label{final}
Let $(M,g)$ be a {\em cosmological spacetime}, that is, 
\begin{itemize}
    \item[i)] $(M,g)$ is globally hyperbolic with compact Cauchy hypersurfaces, and
    \item[ii)] $Ric (v,v) \geq 0$, for all $v \in TM$ timelike. 
\end{itemize}
Assume that $(M,g)$ is causally geodesically complete and that it admits a temporal function $f$ which satisfies properties $(a)$ and $(b)$ in Theorem \ref{mainthm}. Then $(M,g)$ is isometric to $(\mathbb{R}\times \Sigma, -dt^2 + h_0)$, where $(\Sigma, h_0)$ is a compact Riemannian manifold.
\end{cor}
\begin{proof}
Just note that the causal geodesic completeness and global hyperbolicity imply that $(M,g)$ is timelike Cauchy complete by Theorem \ref{complete}, and since the Cauchy hypersurfaces are compact, every level hypersurface of $f$ must be Cauchy hypersurfaces as well, and hence compact. Thus, all the hypotheses in Theorem \ref{mainthm} are in force, and we conclude that $(M,g)$ satisfies the NOH. The splitting result for a cosmological spacetime with the NOH has been proven by Galloway in \cite{G1}. 
\end{proof}

\section{Applications beyond the timelike foliation}\label{sec2.4}

In this section we are going to show that our technique provides results even when the timelike foliation is not explicitly given. 

Let $(M,g)$ be a timelike geodesically complete globally hyperbolic spacetime with a Cauchy hypersurface $\Sigma$. Consider all the inextendible unit future timelike geodesics $\{\gamma_x\}_{x\in \Sigma}$ of $(M,g)$ passing orthogonally through $\Sigma$ at $t=0$. Let $F:\mathbb{R}\times \Sigma\rightarrow M$, $F(t,x):=\exp_{\perp}^{\Sigma}(t\cdot \dot{\gamma}_x(0))=\gamma_x(t)$, be the onto smooth map naturally provided by these geodesics. Let us endow ${\mathbb R}\times \Sigma$ with the pull-back $F^*g$, which can be written as
\begin{equation}\label{kg}
	F^*g=-dt^2 + h_{t},
\end{equation}
where $h_{t}$ is a $1$-parameter family of (possibly degenerate) symmetric semi-definite positive $(2,0)$-tensors on $\Sigma$ (in fact, the eventual presence of focal points for the geodesics passing orthogonally through $\Sigma$ will prevent $F$ from being an immersion, in general). 

The key observation here consists of realizing that, even if $F$ is not a diffeomorphism and $F^*g$ is not a metric, the pair $(\mathbb{R}\times \Sigma,-dt^2 + h_{t})$ can be naturally endowed with a chronology relation that captures the essential information needed from $(M,g)$ to replicate the arguments from previous sections. As we will see below, this will be done by combining via $F$ the chronological structure in $(\mathbb{R}\times \Sigma,-dt^2 + h_{t})$ with the geometric objects of $(M,g)$. 

First, note that
\begin{equation}\label{du}
h_t(v,w)=\langle J_v(t),J_w(t) \rangle\qquad\forall v,w\in T_x\Sigma,\quad\forall t\in{\mathbb R},
\end{equation}
where $J_u$ is the Jacobi vector field along the geodesic $\gamma_x$ in $(M,g)$ with $J_u(0)=u$, $\dot{J}_u(0)=-L(u)$, and being $L$ the Weingarten shape operator of $\Sigma$ in $(M,g)$.

A smooth curve $\alpha:I\rightarrow {\mathbb R}\times \Sigma$, $\alpha(s)=(t(s),x(s))$, is said to be {\em future [past] timelike} in $({\mathbb R}\times \Sigma,-dt^2+h_t)$ if 
\[
\dot{t}(s)>\sqrt{h_{t(s)}(\dot{x}(s),\dot{x}(s))},\qquad [\;\; \dot{t}(s)<-\sqrt{h_{t(s)}(\dot{x}(s),\dot{x}(s))}\;\; ]\qquad\hbox{for all $s\in I$} 
\]
(note that the eventual degeneracy of $h_t$ does not affect this definition). Then, we define the notion of {\em chronology} $\ll$ in $({\mathbb R}\times \Sigma,-dt^2+h_t)$ by using timelike curves as usual. In particular, the concept of omniscient vector field can be also define in $({\mathbb R}\times \Sigma,-dt^2+h_t)$.

Suppose now the existence of a positive function $\mathfrak{a}:[0,\infty)\rightarrow \mathbb{R}$, with $\dot{\mathfrak{a}}(t_0)>0$ for some $t_0\in [0,\infty)$, 
such that
\begin{equation}\label{hyp}
	\sqrt{n}\langle R(\dot{\gamma}(t),v(t))\dot{\gamma}(t),w(t)\rangle\geq -\frac{\ddot{\mathfrak{a}}(t)}{\mathfrak{a}(t)}\;\;\;\forall t\in [t_0,\infty)\;\quad\hbox{and}\;\quad \int_0^{\infty}\frac{dt}{\mathfrak{a}(t)}=\infty,
\end{equation}
where $v(t), w(t)\in \dot{\gamma}(t)^{\perp}\cap \hat{T}M$ for all $t$, and $\gamma$ is any inextendible unit future timelike geodesic $\gamma$ passing orthogonally through $\Sigma$ at $t=0$. Then (recall the proof of Lemma \ref{tamee}), there exists some $\epsilon>0$ small enough such that 
\[
\epsilon \sqrt{h_t(J_u(t),J_u(t))} \leq \mathfrak{a}(t) \cdot \sqrt{h_0(u,u)}, \quad \forall u\in T\Sigma,\; \forall t\in [0,\infty),
\]
which in virtue of (\ref{du}) translates into 
$$\sqrt{h_{t}(u,u)} \leq \mathfrak{a}(t)\epsilon^{-1} \sqrt{h_{0}(u,u)}, \quad \forall u\in T\Sigma,\; \forall t\in [0,\infty).$$	
But this last inequality (compare with the notion of future tame) implies that $\partial_t$ is future omniscient in $({\mathbb R}\times \Sigma,-dt^2+h_t)$ (recall the corresponding argument in the proof of Theorem \ref{main1}), and an immediate adaptation of the argument in the proof of Proposition \ref{omnitonoh} allows us to conclude that the original spacetime $(M,g)$ satisfies the NOH condition. Summarizing, the following result holds:

\begin{thm}\label{final'}
	Let $(M,g)$ be a {\em cosmological spacetime}, that is, 
	\begin{itemize}
		\item[i)] $(M,g)$ is globally hyperbolic with a compact Cauchy hypersurface $\Sigma$, and
		\item[ii)] $Ric (v,v) \geq 0$, for all $v \in TM$ timelike. 
	\end{itemize}
	Assume that $(M,g)$ is timelike geodesically complete and satisfies the curvature constraint (\ref{hyp}).
	Then $(M,g)$ stisfies the NOH condition, and thus, it is isometric to $(\mathbb{R}\times \Sigma, -dt^2 + h_0)$, where $(\Sigma, h_0)$ is a compact Riemannian manifold. \qed
\end{thm}


\section*{Acknowledgments}

The authors would like to acknowledge the enormous and deep contributions by Penrose to the General Theory of Relativity, which continue being inexhaustible sources of inspiration for the specialists in the field. The first two authors are partially supported by the grant PID2020-118452GBI00
(Spanish MICINN). The last two authors are partially supported by the grant PY20-01391 (PAIDI 2020).


\end{document}